\documentclass[letterpaper, 10 pt, conference]{ieeeconf}  

\IEEEoverridecommandlockouts                              
\usepackage{graphics} 
\usepackage{epsfig} 
\usepackage{mathptmx} 
\usepackage{times} 
\usepackage{amsmath} 
\usepackage{amssymb}  
\newtheorem{theorem}{Theorem}
\usepackage{epstopdf}

\usepackage{cite}
\usepackage{epsfig}
\title{\LARGE \bf
Optimal Relay Selection with Non-negligible Probing Time
}

\author{Yang Liu$^{*}$, Yi Ouyang$^{*}$ and Mingyan Liu
\thanks{Authors are from Department of Electrical Engineering and Computer Science, University of Michigan, Ann Arbor. Contact information : \{youngliu, ouyangyi, mingyan\}@umich.edu}
\thanks{$^*$ indicates equal contribution.}%
}

\begin{document}

\maketitle
\thispagestyle{empty}
\pagestyle{empty}

\begin{abstract}

In this paper an optimal relay selection algorithm with non-negligible probing time is proposed and analyzed for cooperative wireless networks.
Relay selection has been introduced to solve the degraded bandwidth efficiency problem in cooperative communication. Yet complete information of relay channels often remain unavailable for complex networks which renders the optimal selection strategies impossible for transmission source without probing the relay channels. 
Particularly when the number of relay candidate is large,
even though probing all relay channels guarantees the finding of the best relays at any time instant, the degradation of bandwidth efficiency due to non-negligible probing times, which was often neglected in past literature, is also significant. 
In this work, a stopping rule based relay selection strategy is determined for the source node to decide when to stop the probing process and choose one of the probed relays to cooperate with under wireless channels' stochastic uncertainties. This relay selection strategy is further shown to have a simple threshold structure. At the meantime, full diversity order and high bandwidth efficiency can be achieved simultaneously. Both analytical and simulation results
are provided to verify the claims.

\end{abstract}

\section{Introduction}

There is a quickly increasing demand for high data rates in wireless communication along with the skyrocketing usage of mobile devices. Increasing transmission diversities is among the most promising techniques and is attracting much attention \cite{zhang2007high,heath2005switching,lehmann2007evaluation}.
Transmission diversity in communication systems provides more than one copy of the transmitted signal
to the destination node, with which the destination can decode the transmitted signal even if some copies of the signal are distorted due to the time varying nature of transmission channels; therefore the system performance can be expected to improve significantly. More precisely the diversity order of a communication system can be measured by the relationship between the error probability, denoted by $\mathbf{P_e}$, and the Signal-Noise-Ratio (SNR) as follows :
a system has diversity order $d$ if
\begin{equation}
\mathbf{P_e} = O(\text{SNR}^{-d}).
\end{equation}

Cooperative communication techniques have been introduced to increase system's diversity order (see \cite{nosratinia2004cooperative}).
In a cooperative network, when a node receives a packet not destined for it, instead of simply discarding the packet it can choose to help to relay and via such help, the source transmission can improve its diversity by sending signal through the relay channels.
Although cooperative communications look promising, the gain is not immediate as cooperation incurs a cost of wireless resources such as frequency, active air time, and power resources to enable relaying.
This trade-off between gaining diversity and conserving the wireless resources is mainly two-fold.

The first aspect is due to signal transmission.
Since wireless resources are spent for relaying the signal to the destination, more bandwidth is needed when more relay nodes are involved in a transmission.
To see this point more clearly, if $N$ nodes are relaying the signal, the bandwidth efficiency degrades to $\frac{1}{N+1}$ dues to the fact that there are all-together $N$ relays and the source are transmitting the same copy of signal. This problem of degraded bandwidth efficiency is solved by the introduction of several relay selection protocols as detailed in \cite{laneman2004cooperative,bletsas2006simple,zhao2006improving,krikidis2008amplify,ibrahim2008cooperative,vicario2009opportunistic, liu2009cooperative,jing2009single,zou2010adaptive,hong2006energy,li:TWC11,
jamal2011interference,adebo:milcom14}. With relay selection, only one relay instead of all is selected for the purpose of cooperation.
Therefore, the bandwidth efficiency is leveraged to $\frac{1}{2}$ instead of $\frac{1}{N+1}$ which decreases as the number of relays increases.
A relay selection protocol is called to achieve full diversity order if the resulting diversity order is $N+1$ for a network with $N$ relay nodes.
When complete relay channel information is available,
relay selection protocols with different selection metrics have been proved to achieve full diversity order in the literature (see \cite{laneman2004cooperative,bletsas2006simple,zhao2006improving,krikidis2008amplify,ibrahim2008cooperative,vicario2009opportunistic, liu2009cooperative,jing2009single,zou2010adaptive}).
For example in \cite{jing2009single}, best relay selection (choosing the relay with best channel SNR), best worse channel selection (choosing the relay with best worse channel condition) and best harmonic mean method have been proved to achieve full diversity order.

%

The second aspect, which was often neglected, comes from channel probing. Due to the dynamic nature of channel conditions, channel probing is needed for exploring each relay channel's instantaneous transmission quality in order to find the best relay to cooperate with. Specifically we consider the following channel probing procedure. Before each transmission, the source sequentially probes the channels between the source and relays, and the channels between relays and the destination. The source can stop the probing process at any time, and select one of the probed relays to cooperate with. For each probing, a carrier sensing packet is sent to reveal each relay channel's instantaneous quality  (see \cite{kanodia2004moar,liu2006sensing}).
Despite the efforts towards reducing the size of a sensing packet, the probing time remains non-negligible.
In this regards, even though probing all channels guarantees the discovery of the set of best relays, the degradation of bandwidth efficiency due to probing times is significant, especially when the number of relays is large. Therefore opportunistic channel probing (as commonly defined in the Opportunistic Spectrum Access (OSA) \cite{survey:osa}) is needed for a bandwidth efficient system.

In this work we design a bandwidth efficient relay selection strategy in a relay network with non-negligible probing time. Technically we adopt theory of optimal stopping rule (which has been previously applied to opportunistic spectrum access, for instance \cite{zheng2009distributed,liu2013stay} and references therein.) to solve our relay selection problem. Our main contributions are as follows.
\begin{itemize}
\item We determine the optimal relay selection strategy and we show the optimal strategy is a simple threshold enabled stopping rule.

\item Our relay selection strategy achieves full diversity order.

\item This optimal relay selection results in bounded probing time, which further implies a relay selection scheme with high bandwidth efficiency.
\end{itemize}

The rest of our paper is organized as follows.
We formulate the relay selection problem as an optimal stopping problem in Section \ref{sec:problem}.
In Section \ref{sec:solution}, we solve the optimal stopping problem and define the corresponding relay selection strategy.
We present in Section \ref{sec:performance} the analysis of diversity order and bandwidth efficiency of the proposed strategy.
In Section \ref{sec:simulation}, we verify our results through simulation and conclude our paper in Section \ref{sec:conclusion}.


\section{System Model and Problem Formulation}
\label{sec:problem}
\subsection{System model}
We consider a relay network with one source $S$, one destination $D$ and $N$ relays $R_1,R_2,...,R_N$ (as depicted in Fig.\ref{system1}) . We model the channels between any two nodes ($S$, $D$ or $R_1,R_2,\dots,R_N$) as discrete time independent Rayleigh fading channels.
Specifically the direct source to destination channel, the $N$ source to relay channels, and the $N$ relay to destination channels are modeled as follows
	\begin{align}	
	&y_{s,d} = h_{s,d}\sqrt{P_s}x_{s}+ \eta_{s,d},\\	
	&y_{s,n} = h_{s,n}\sqrt{P_s}x_{s}+\eta_{s,n}\text{, for }n = 1,2,...,N,\\
    &y_{n,d} = h_{n,d}\sqrt{P_r}x_{n}+\eta_{n,d}\text{, for }n = 1,2,...,N,
	\end{align}
	where $P_s$  and $P_r$ are the transmission power at the source and relays respectively; $x_{s}, x_{n}$ are the signals transmitted or retransmitted by the source and the $n$th relay respectively;
	$y_{s,d}, y_{n,d}, y_{s,n}, n=1,2,\dots,N$ are the received signals at the destination $D$ from sources and $N$ relays, and the $n$th relay $R_n,n=1,2,\dots,N$, respectively. The channel noises $\eta_{s,d}, \eta_{s,n}, \eta_{n,d}$ are modeled as Gaussian random variables with zero-mean and variance $\eta_0$.
	\begin{figure}[h]
    \centering
        \includegraphics[width=0.2\textwidth]{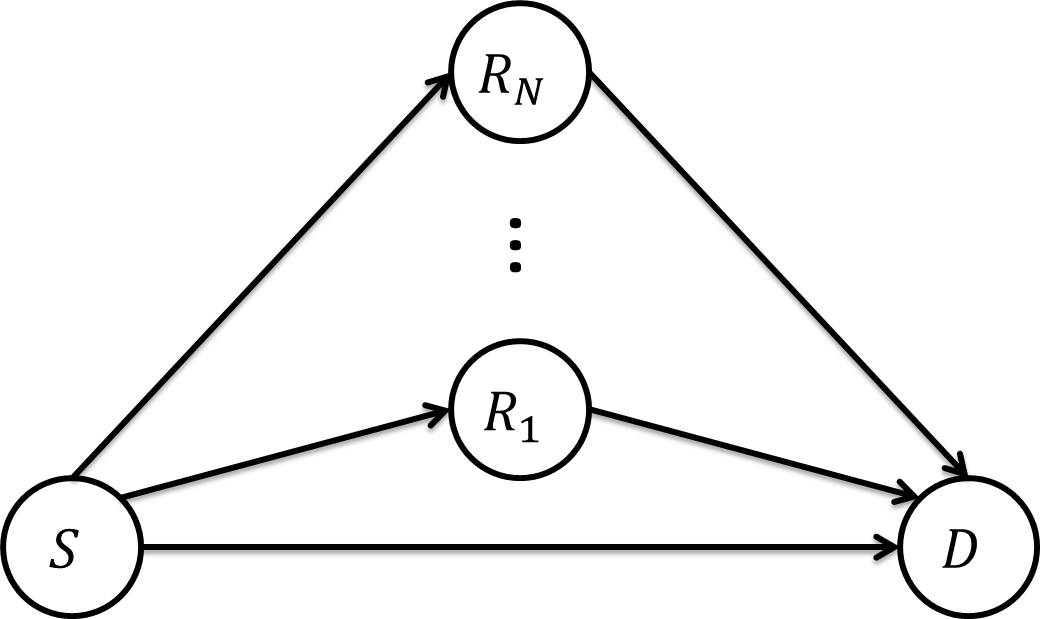}
    \caption{The relay network}
\label{system1}
\end{figure}
	
For all the channels, $h_{s,d},h_{s,n},h_{n,d},n=1,2,\dots,N$ are coefficients that capture the effect of path loss, fading etc.
We assume these channel coefficients remain the same during the transmission of one source signal.
Moreover, the channel coefficients $h_{s,d},h_{s,n},h_{n,d},n=1,2,\dots,N$ are modeled as zero-mean, complex Gaussian random variables with unit variance for each signal transmission. For simplicity of later analysis, we use $\omega_{s,d},\omega_{s,n},\omega_{n,d},n=1,2,\dots,N$ to denote the magnitude of the channel coefficients:
	\begin{align}
		&\omega_{s,d} = |h_{s,d}|^2,\\
		&\omega_{s,n} = |h_{s,n}|^2, n = 1,2,\dots,N,\\
		&\omega_{n,d} = |h_{n,d}|^2, n = 1,2,\dots,N.
	\end{align}
	
For each signal transmission, besides the direct transmission to the destination, the source node selects one relay out of the $N$ candidates to cooperate using perfect Direct Forwarding (DF) protocol (for technical details please refer to \cite{liu2009cooperative}).
That is, when relay $R_n$ is selected, the source $S$ will first send the signal $x_s$ to both the relay $R_n$ and the destination $D$.
Let $\tilde{x}_{s,n}$ denotes the decoded signal at $R_n$, and the transmitted signal of relay $R_n$ is given by
	\begin{align}
			x_{n} = \left\{
			    \begin{array}{rl}
			    x_{s}, & \text{if } x_s = \tilde{x}_{s,n},\\
			    0, & \text{if }x_s \neq \tilde{x}_{s,n}.
			    \end{array} \right.
	\end{align}	
i.e., the relays will forward the message/signal only if it has been correctly decoded.


\subsection{The Relay Selection Problem}

We formulate the relay selection problem with channel probing. For each relay $R_n$ we adopt the following index $\omega_n$ introduced in \cite{liu2009cooperative} as the criteria for selecting relay,
\begin{align}
\omega_n = \frac{2q_1q_2\omega_{s,n}\omega_{n,d}}{q_1\omega_{n,d}+q_2\omega_{s,n}}
\label{eq:omega_n}~,
\end{align}
with $q_1, q_2$ being constants as defined in \cite{liu2009cooperative}.
When complete information (all $\omega_{s,n}, \omega_{n,d}$s) of all channels is available, it is shown in \cite{liu2009cooperative} that selecting the relay with maximum $\omega_n$ gives full diversity gain, and this is also the major reason we adopt $\omega_n$ as the index for relay $R_n,n=1,2,\dots,N$. Moreover if we view $\omega_n$ as an approximated channel gain of using relay $R_n$, $P_s \cdot \omega_n$ becomes the approximated signal power at the destination through $R_n,n=1,\dots,N$.

Before each transmission, we assume the source sequentially probes the channels between the source and
relays, and the channels between relays and the destination. The channel probing procedure, as shown in  Fig. \ref{fig:probing}, can stop at any stage $n\leq N$ when the channels connected to relays $R_1,R_2,\dots,R_n$ are probed and select one relay $R_k,k\leq n$ to cooperate with.
Let $T_s$ be the probing time to probe the channels between $(S,R_n)$ and $(R_n, D)$ (for example, $T_s$ could be a cycle of RTS/CTS period for IEEE 802.11 channel sensing protocol)
and $T_{tran}$ be the time for signal transmission of the source and the selected relay. Then for any $n=1,2,\dots,N$, the time to probe relays $R_1,R_2,\dots,R_n$ is $nT_s$, and the bandwidth efficiency is given by
\begin{eqnarray}
c_n = \frac{T_{tran}}{T_{tran}+nT_s} =\frac{1}{1 + n\tau},
\end{eqnarray}
where $\tau = \frac{T_s}{T_{tran}}$ denoting the ratio between the probing time and the transmission time.

\begin{figure}[h]
    \centering
        \includegraphics[width=0.4\textwidth]{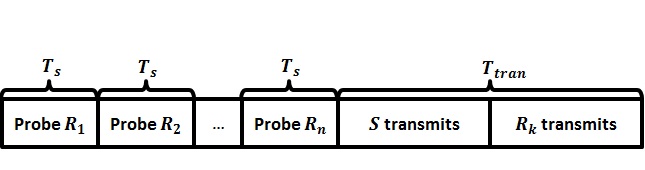}
    \caption{The Channel Probing Procedure}
\label{fig:probing}
\end{figure}

We define the signal to noise ratio (SNR) $\gamma$ of the system to be $\gamma=P/\eta_0$,
where $P$ is the total transmission power (including source and the selected relay). Consider a constant power scheme and for a fair comparison, the sum of transmission power at the source and the selected relay should equal the total power $P$ times the bandwidth efficiency, i.e.
	\begin{align}
	P_s+P_r = P \cdot c_n~.\label{eqn:power}
	\end{align}
The reason we take total transmission power to be $P \cdot c_n$ in \eqref{eqn:power} is due to the fact that $(1-c_n)$ fraction of total power $P$ is used in channel probing. Denote $r$ as the power ratio $r = \frac{P_s}{P \cdot c_n}$. Then,
\begin{align}
&P_s = rP \cdot c_n, \\
&P_r = (1-r)P \cdot c_n.
\end{align}

When the channel probing process stops at stage $n$ and relay channels $\omega_{s,1},\omega_{s,2},\dots,\omega_{s,n}$ are probed, the
maximum relay index is $\max_{k\leq n} \omega_n$, and the corresponding approximated signal power at the destination is given by
\begin{align}
P_s\max_{1 \leq k \leq n}\omega_k= r P\cdot c_n\Omega_n,\label{eq:asp}
\end{align}
where
\begin{align}
\Omega_n = \max_{1\leq k \leq n}\omega_k, n=1,2,\dots,N~.
\end{align}

The relay selection problem defined in this paper is to decide when to stop probing the relay channels and which probed relay to cooperate with to maximize the expected value of the approximated signal power at the destination.\footnote{The objective function defined in \eqref{stoppingproblem_old} is an intuitive measure. However we will show in Section IV the stopping rule maximizing this particular objective achieves full diversity order.}
Formally, we want to choose a stopping time $N_s$ with respect to the sequential channel probings and the channel realizations $\{\omega_{s,n},\omega_{n,d},n=1,2,\dots,N\}$ and we formulate the relay selection problem as an optimal stopping problem as follows.
\begin{align}
\begin{array}{ccll}
&\max_{N_s}& E[rP\cdot c_{N_s} \Omega_{N_s}]\\
&\text{s.t.}& N_s \text{ is a stopping time}, N_s\leq N \quad a.s.
 \end{array}
 \label{stoppingproblem_old}
\end{align}

Assuming $rP$ to be a constant, then \eqref{stoppingproblem_old} is equivalent to the following optimal stopping problem.
\begin{align}
\begin{array}{ccll}
&\max_{N_s}& E[c_{N_s} \Omega_{N_s}]\\
&\text{s.t.}& N_s \text{ is a stopping time}, N_s\leq N \quad a.s.
 \end{array}
 \label{stoppingproblem}
\end{align}

\section{Solution to the Optimal Stopping Problem}
\label{sec:solution}

From the theory of optimal stopping times (see \cite[chap. 2]{uclastopping})
, we define the value functions with the information state $\{\Omega_n, 1\leq n \leq N \}$ as follows.

For stage $n=1,2,\dots,N$, define
\begin{align}
V_n(x) := \max_{N_s\text{ is a stopping time}, n \leq N_s \leq N}\mathbf{E}[c_{N_s}\Omega_{N_s}|\Omega_n = x]~,
\label{eq:valuefunction}
\end{align}
and we can write down the backward induction for the value functions as follows,
\begin{align}
&~~~~~~~~~V_N(x) = c_N \cdot x,\nonumber \\
&V_n(x) =\max\left\{c_n x, \mathbf{E}[V_{n+1}(\Omega_{n+1})|\Omega_n=x] \right\}, n \leq N-1.
\label{eq:dynamicprogram}
\end{align}
From the backward induction, we obtain the structure of the optimal stopping time stated in the following theorem.
\begin{theorem}
\label{thm:stoppingrule}
The optimal stopping rule $N_s^*$ for the optimal stopping problem described by (\ref{stoppingproblem}) is given by thresholds $t_1,t_2,\dots,t_{N-1}$ such that
\begin{align}
N_s^* = \inf\{n\geq1 : c_n \Omega_n \geq t_n \},
\label{eq:optimalstopping}
\end{align}
where the threshold $t_n,~n=1,2,\dots,N-1$ is the unique solution of the following fixed point equations
\begin{align}
c_{n} t_n  = \mathbf{E}[V_{n+1}(\Omega_{n+1})|\Omega_{n}=t_n].
\label{eq:threshold}
\end{align}
Consequently, the value functions satisfy
\begin{align}
V_{n}(x) = \left\{ \begin{array}{ll}
	c_{n} x & \text{ if }x \geq t_n, \\
	\mathbf{E}[V_{n+1}(\Omega_{n+1})|\Omega_{n}=x] & \text{ if }x < t_n.
\end{array}\right.
\label{eq:valuefunction_th}
\end{align}

\end{theorem}
\begin{proof}
The proof can be found in Appendix-\ref{pf:stoppingrule}.
\end{proof}

Theorem \ref{thm:stoppingrule} states that the optimal stopping rule $N_s^*$ to problem (\ref{stoppingproblem}) we formulated in Section \ref{sec:problem} is described by a set of thresholds $t_1,t_2,\dots, t_{N-1}$, based on which we propose the relay selection strategy $\text{RS\_OSR} =\{d_1^*,d_2^*,\dots,d_N^*\}$ as follows.
The decision $d_n^*$ at each stage $n\leq N-1$ is given by
\begin{align}
d_n^* = \left\{
\begin{array}{ll}
\text{Stop and choose }R_k & \text{ if }c_n\Omega_n \geq t_n, \omega_k = \Omega_n, \\
\text{Continue } & \text{ if }c_n\Omega_n < t_n,
\end{array}\right.
\label{eq:strategyn}
\end{align}
and at the final stage $N$
\begin{align}
d_N^* = \left\{
\begin{array}{ll}
\text{Choose }R_k & \text{ if }\Omega_N \geq \omega_{s,d}, \omega_k = \Omega_n, \\
\text{Do not choose any relay } & \text{ if }\Omega_N < \omega_{s,d}.
\end{array}\right.
\label{eq:strategyfinal}
\end{align}

\section{Performance Analysis}
\label{sec:performance}
In this section, we analyze the performance of the relay selection strategy RS\_OSR defined in Section \ref{sec:solution}.
We show in Section \ref{sub:errorprob} that RS\_OSR achieves full diversity order.
In Section \ref{sub:bandwidth} we obtain an upper bound on the expected stopping time for RS\_OSR.
\subsection{Diversity Gain}
\label{sub:errorprob}
To compute the diversity gain of the relay selection strategy RS\_OSR, we
consider the error probability and show that RS\_OSR achieves full diversity order.
This result is stated in the theorem below.
\begin{theorem}
\label{thm:fullorder}
The relay selection strategy RS\_OSR defined by (\ref{eq:strategyn})-(\ref{eq:strategyfinal}) achieves full diversity order.
That is, when $N$ relays are available
\begin{align}
\text{Diversity gain}
:= &- \lim_{\gamma \rightarrow \infty} \frac{\log(\mathbf{P}_e(\gamma))}{\log(\gamma)} = (N+1),
\end{align}
where $\mathbf{P}(\gamma)$ is the error probability under RS\_OSR with $\gamma$ being the SNR.
\end{theorem}

\begin{proof}
In order to compute the error probability of our relay selection strategy, we first consider the error probability when the relay selection stops at stage $n$ and $R_k, ~k \leq n$ is selected.
Let $\mathbf{P}_{e,n}(\gamma,x)$ be the error probability conditional on $\Omega_n=x$ and the relay selection stops at stage $n$. From Appendix \ref{pf:fullorder} there exists constants $A_M, B_M$ such that
\begin{align}
 &\mathbf{P}_{e,n}(\gamma,x) \leq A_M\exp(-B_M\frac{1-r}{2q_1}xc_n\gamma) \nonumber \\
 &~~~+ A_M\exp(-B_M\frac{r}{2q_2} xc_n\gamma) := Q_n(\gamma,x)~.
\label{eq:upperboundn}
\end{align}
Now using (\ref{eq:upperboundn}), we try to get an upper bound for the error probability when the relay selection stop at stage $n$.
Under RS\_OSR, (\ref{eq:strategyn}) implies that
the relay selection stop at stage $n, ~n\leq N-1$ if $\Omega_n \geq t_n/c_n$ and $R_k,~k\leq n$ is selected when $\omega_k = \Omega_n$.
Consequently, for $n \leq N-1$ we have
\begin{align}
&\mathbf{P}(\text{error, relay selection stops at time }n) \nonumber\\
= & \mathbf{P}(\text{error}, \Omega_n \geq t_n/c_n) \nonumber\\
= & \int_{t_n/c_n}^{\infty}\mathbf{P}_{e,n}(\gamma,x)dF_{\Omega_n}(x) \nonumber\\
\leq & \int_{t_n/c_n}^{\infty}Q_n(\gamma,x)dF_{\Omega_n}(x) \nonumber\\
\leq & Q_n(\gamma,t_n/c_n).
\label{eq:upperbound_n}
\end{align}
By (\ref{eq:strategyfinal}), RS\_OSR selects relay $R_k$ if $\omega_k = \Omega_N \geq \omega_{s,d}$
and selects no relay if $\Omega_N < \omega_{s,d}$ at stage $N$. From \cite[chap. 8]{liu2009cooperative}, we know that the error probability if the selection stops at stage $N$ is upper bounded by
\begin{align}
&\mathbf{P}(\text{error, relay selection stops at time }N) \nonumber\\
\leq & (CG c_N\gamma)^{-(N+1)}\mathbf{P}(\text{relay selection stops at time }N) \nonumber\\
\leq & (CG c_N\gamma)^{-(N+1)},
\label{eq:upperbound_N}
\end{align}
where $CG$ is a constant defined in \cite[chap. 8]{liu2009cooperative}.

The combination of (\ref{eq:upperbound_n}) and (\ref{eq:upperbound_N}) gives the complete error probability as follows.
\begin{align}
\mathbf{P_e}(\gamma)
=& \sum_{n=1}^N\mathbf{P}(\text{error, relay selection stops at time }n) \nonumber\\
\leq & \sum_{n=1}^{N-1}Q_n(\gamma,t_n/c_n)+(CG c_N\gamma)^{-(N+1)} \nonumber\\
= &(1+o(\gamma)) (CG c_N\gamma)^{-(N+1)}~,
\label{eq:errorrate}
\end{align}
where the last equality in (\ref{eq:errorrate}) is true because $Q_n(\gamma,t_n/c_n)$ is the sum of two exponential functions that are dominated by $(CG c_N\gamma)^{-(N+1)})$ when $\gamma$ is large.
From the upper bound (\ref{eq:errorrate}) the diversity gain can be computed by
\begin{align}
\text{Diversity gain}= &- \lim_{\gamma \rightarrow \infty} \frac{\log(\mathbf{P_e}(\gamma))}{\log(\gamma)} \nonumber\\
\geq &- \lim_{\gamma \rightarrow \infty} \frac{\log((1+o(\gamma)) (CG c_N\gamma)^{-(N+1)})}{\log(\gamma)} \nonumber\\
= & (N+1).
\label{eq:diversitygain}
\end{align}
\end{proof}
Theorem \ref{thm:fullorder} states that the relay selection strategy RS\_OSR achieves the full diversity order of $N+1$ with $N$ relays.

\subsection{Bandwidth Efficiency}
\label{sub:bandwidth}

Let $N^*_{s}$ be the optimal stopping time in RS\_OSR (defined by (\ref{eq:optimalstopping})).
The expectation of $N^*_s$ can be computed as follows
\begin{align}
	\mathbf{E}[N_s^*] =& 1 + \sum_{n = 1}^N P( N_s^* > n) = 1+\sum_{n = 1}^N \mathbf{P}(\cap_{k=1}^{n} \{ \Omega_k < t_{k}\}) \nonumber\\
			  \leq & 1+\sum_{n = 1}^N \mathbf{P}( \Omega_n < t_n) = 1+\sum_{n = 1}^N (\mathbf{P}(\omega_1 < t_n))^n.
			  \label{eq:NsUB1}
\end{align}
Let $t_{N,max} = \max_{n=1,2,\dots,N}(t_n)$ and $p_N=\mathbf{P}(\omega_1 < t_{N,max})$, then from (\ref{eq:NsUB1}) we further obtain
\begin{align}
	\mathbf{E}[N_s^*] \leq & 1+ \sum_{n = 1}^N (\mathbf{P}(\omega_1 < t_n))^n \nonumber\\
				    \leq & \sum_{n = 0}^N p_N^n = \frac{p^{N+1}_N}{1-p_N}.
			  		\label{eq:ENUB}
\end{align}
If $1-p_N \geq \epsilon$ for some $\epsilon >0$, the expected stopping time of RS\_OSR is bounded for any number of relays.
We show in the theorem below that it is indeed the case.
\begin{theorem}
\label{thm:bandwidth}
There exists some positive constant $\epsilon > 0$ such that $1-p_N \geq \epsilon$ for any $N$.
Consequently, the expected number of probing for the relay selection strategy RS\_OSR described by (\ref{eq:strategyn})-(\ref{eq:strategyfinal}) is bounded above by
\begin{align}
	\mathbf{E}[N_s^*] \leq & \frac{1-\epsilon}{\epsilon}~.
\end{align}

\end{theorem}
\begin{proof}
Since $V_N(x) = c_Nx$, from (\ref{eq:threshold}) we know that $t_{N-1}$ satisfies
\begin{align}
c_{N-1}t_{N-1} = &\mathbf{E}[V_{N}(\Omega_{N})|\Omega_{N-1} = t_{N-1}] \nonumber\\
               = &\mathbf{E}[c_{N}\max(t_{N-1},\omega_{N})].
\end{align}
Let $h(x) = \mathbf{E}[\max(x,\omega_{1})]$, we get
\begin{align}
\frac{h(t_{N-1})}{t_{N-1}} = \frac{c_{N-1}}{c_N} =\frac{1+N\tau}{1+(N-1)\tau}> 1.
               \label{eq:tN-1}
\end{align}
Note that the function $\frac{h(x)}{x}$ is strictly decreasing as shown below.
We first compute the derivative of $h(x)$ as follows.
\begin{align}
h'(x) =
&\left(\int_{0}^x x dF_{\omega_1}(\omega)+ \int_{x}^{\infty} \omega dF_{\omega_1}(\omega) \right)' \nonumber\\
= & x F'_{\omega_1}(x) + F_{\omega_1}(x)
- xF'_{\omega_1}(x) \nonumber\\
= & F_{\omega_1}(x),
\end{align}
where $F_{\omega_1}(x)$ is the CDF of $\omega_1$.
Moreover, $h(x) = \mathbf{E}[\max(x,\omega_{1})] \geq x$. Then
\begin{align}
\left(\frac{h(x)}{x}\right)' = &\frac{1}{x^2}\left(xh'(x) - h(x)\right) \nonumber\\
= &\frac{1}{x^2}\left(xF_{\omega_1}(x) - h(x)\right) \nonumber\\
\leq &\frac{1}{x^2}\left(xF_{\omega_1}(x) - x\right) <0.
\end{align}
Define $g(x)$ as the inverse function of $\frac{h(x)}{x}$, then $g(x)$ is also strictly decreasing.
\\
Let $t^*$ be the solution to $\frac{h(x)}{x} = 1$, i.e. $t^* = g(1)$.
Then, it follows form (\ref{eq:tN-1}) that
\begin{align}
t_{N-1} = g\left(\frac{c_{N-1}}{c_N} \right) < g(1) = t^*.
\end{align}
Furthermore, we want to show that
\begin{align}
t_n \leq t_{N-1} \text{ for all }n\leq N-1.
\label{eq:tnleqtN}
\end{align}
The proof is done by induction. \\
(\ref{eq:tnleqtN}) is true for $N-1$. Assume the (\ref{eq:tnleqtN}) holds for $n+1$.\\
For stage $n$, if $t_n \leq t_{n+1}$, we get $t_n \leq t_{n+1} \leq t_{N-1}$ by the induction hypothesis.
If $t_n > t_{n+1}$, $\max(t_n,\omega) \geq t_{n+1}$ for any $\omega$.
From (\ref{eq:valuefunction_th}) we obtain
\begin{align}
\mathbf{E}[V_{n+1}(\Omega_{n+1})|\Omega_n = t_n]
 = &\mathbf{E}[V_{n+1}(\max(t_n,\omega_{n+1}))] \nonumber\\
 = &\mathbf{E}[c_{n+1}\max(t_n,\omega_{n+1})].
\end{align}
Then, from (\ref{eq:threshold}) for $t_n$ we get
\begin{align}
c_nt_n = &\mathbf{E}[V_{n+1}(\Omega_{n+1})|\Omega_n = t_n] \nonumber\\
	   = &\mathbf{E}[c_{n+1}\max(t_n,\omega_{n+1})]  \nonumber\\
	   = &c_{n+1}h(t_n).
\end{align}
Therefore,
\begin{align}
t_n = g\left(\frac{c_n}{c_{n+1}}\right).
\end{align}
Note that
\begin{align}
\frac{c_n}{c_{n+1}} = \frac{1+(n+1)\tau}{1+n\tau} \geq \frac{1+N\tau}{1+(N-1)\tau}=\frac{c_{N-1}}{c_N}.
\end{align}
Since $g(x)$ is decreasing, we have
\begin{align}
t_n = g\left(\frac{c_n}{c_{n+1}}\right) \leq g\left(\frac{c_{N-1}}{c_{N}}\right) = t_{N-1}.
\end{align}
As a result of the above analysis, we have
\begin{align}
t_n \leq t_{N-1} < t^*
\end{align}
for any $n = 1,2,\dots,N-1$. Moreover,
\begin{align}
p_N=\mathbf{P}(\omega_1 < t_{N,max}) \leq \mathbf{P}(\omega_1 < t^*) < 1.
\end{align}
We can now define $\epsilon = 1-\mathbf{P}(\omega_1 \leq t^*)>0$, then
$1-p_N \geq  1-\mathbf{P}(\omega_1 \leq t^*) =\epsilon$.
From (\ref{eq:ENUB}) we obtain
\begin{align}
	\mathbf{E}[N_s^*] \leq & \frac{p^{N+1}_N}{1-p_N} \leq \frac{1-\epsilon}{\epsilon}
\end{align}
for any total number $N$ of relays.

\end{proof}
From Theorem \ref{thm:bandwidth}, we know that the expected stopping time of the relay selection strategy RS\_OSR is bounded. Therefore, the expected time for each transmission is bounded by
\begin{align}
\mathbf{E}[T_{tran}(1+N_s^*\tau)] \leq T_{tran}(1+\tau\frac{1-\epsilon}{\epsilon}).
\end{align}
\section{Simulation results}
\label{sec:simulation}

We simulate a relay networks and present simulation results. To implement RS\_OSR, we solved the fixed point equation (\ref{eq:threshold}) in an offline fashion. In particular we use an iterative heuristic to solve the equation with the help of Monte Carlo sampling (to remedy the computation complexity from calculating the conditional expectation in (\ref{eq:threshold})).

In the simulation we compare our proposed relay selection strategy RS\_OSR with
the optimal relay selection (w.r.t. indices $\omega_n,n=1,2,\dots,N$ defined by (\ref{eq:omega_n})) that probes all relays, which is denoted by RS\_ALL. Note that the optimality (can achieve full diversity order) of RS\_ALL is proved in \cite{liu2009cooperative}.


We start with comparing the error probability $\mathbf{P_e}$. Fig.\ref{fin1} and Fig.\ref{fin2} show the comparison for error probability with $\tau = 0.05$ and $\tau = 0.1$ respectively. From Fig.\ref{fin1} and Fig.\ref{fin2} we observe that RS\_OSR achieves full diversity order as comparable with RS\_ALL. Moreover our algorithm outperforms RS\_ALL at finer degree consistently. This dues to the gain of efficient bandwidth and power saving. We also observe that the advantages of RS\_OSR is more obvious when $\tau$ is higher. This is intuitively true : our optimal stopping selection strategy helps save sensing time and power and therefore the advantage becomes more and more clear when the sensing complexity becomes higher.

\begin{figure}[h]
    \centering
        \includegraphics[width=0.45\textwidth]{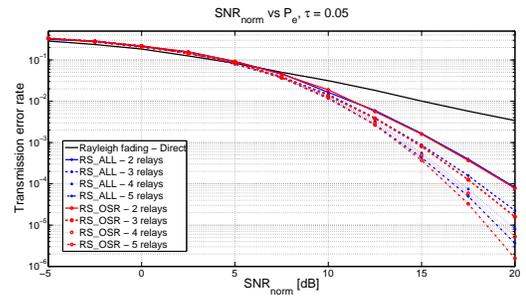}
    \caption{BER performance with $\tau = 0.05$}
\label{fin1}
\end{figure}
\begin{figure}[h]
    \centering
        \includegraphics[width=0.45\textwidth]{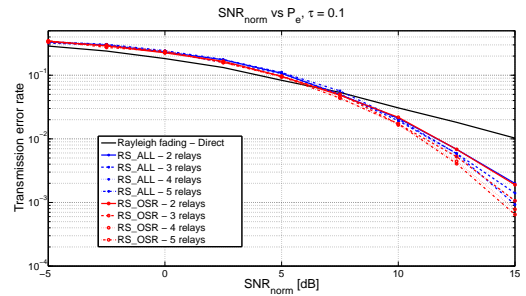}
    \caption{BER performance with $\tau = 0.1$}
\label{fin2}
\end{figure}

Bandwidth efficiency results (rate is measured by the average $c_n$ over all sample path) are shown in Fig.\ref{rate1} and Fig.\ref{rate2}, with $\tau = 0.05$ and $\tau = 0.1$ respectively. From Fig.\ref{rate1} and Fig.\ref{rate2} we conclude that under both cases, the probing time of RS\_OSR is bounded, while the probing time of RS\_ALL increases (thus the bandwidth efficiency decreases) as the number of relays increases.
\begin{figure}[h]
    \centering
        \includegraphics[width=0.45\textwidth]{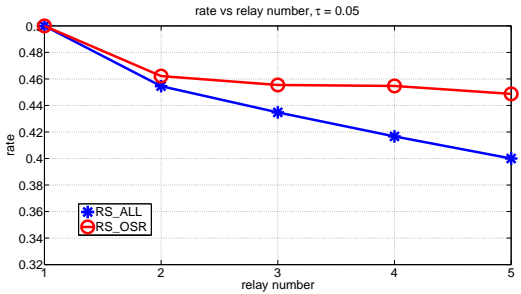}
    \caption{Bandwidth efficiency with $\tau = 0.05$}
\label{rate1}
\end{figure}
\begin{figure}[h]
    \centering
        \includegraphics[width=0.45\textwidth]{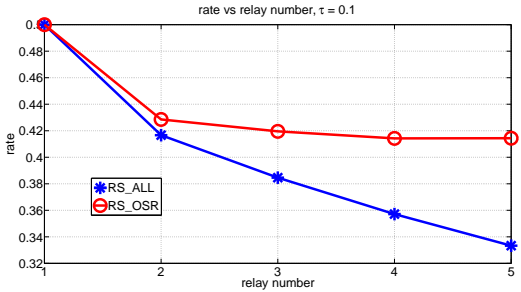}
    \caption{Bandwidth efficiency with $\tau = 0.1$}
\label{rate2}
\end{figure}



\section{Conclusion}
\label{sec:conclusion}
In the presence of non-negligible probing time for getting relay channels' instantaneous quality, obtaining full information for the purpose of relay selection leads to an inefficient use of bandwidth in cooperative communications.
We design and implement a stopping rule based relay selection strategy RS\_OSR and proved its optimality regarding achieving full diversity order.
Moreover, the probing time for sensing relay channels under RS\_OSR is shown to remain bounded regardless of the number of relay candidates. We establish and demonstrate the above two properties by both analytical and simulation results.



\begin{small}
\bibliographystyle{ieeetr}
\bibliography{ref}
\end{small}
\begin{appendices}

\section{Proof of Theorem \ref{thm:stoppingrule}}\label{pf:stoppingrule}

\begin{proof}
First, we want to prove that the value function $V_n(x)$ and $E[V_{n+1}(\Omega_{n+1})|\Omega_{n}=x] $ are convex and increasing. The proof is done by induction.

At the final stage $N$, the value function $V_N(x) = c_N x$ is obviously convex (linear) and increasing.
Assume that $V_{n+1}(x)$ is convex and increasing.
Then, at stage $n$ we have
\begin{align}
\mathbf{E}[V_{n+1}(\Omega_{n+1})|\Omega_{n}=x]
= &\mathbf{E}[V_{n+1}(\max(x,\omega_{n+1}))|\Omega_{n}=x] \nonumber\\
= &\mathbf{E}[V_{n+1}(\max(x,\omega_{n+1}))],
\end{align}
where the last equality holds because $\omega_{n+1}$ is independent of $\Omega_{n}$ (due to the assumption of independence among channels).
Since $V_{n+1}(x)$ is convex and increasing and the function $\max(x,\omega_{n+1})$ is convex and increasing in $x$,
by the property of convex functions (see \cite{boyd2004convex}) we know that their composition
$V_{n+1}(\max(x,\omega_{n+1}))$ is convex and increasing in $x$ for every $\omega_{n+1}$.
Therefore, $\mathbf{E}[V_{n+1}(\Omega_{n+1})|\Omega_{n}=x]$ is also convex and increasing, as it is the
the expectation of $V_{n+1}(\max(x,\omega_{n+1}))$ over $\omega_{n+1}$.
Since $V_{n}(x) =\max\left\{c_{n} x, \mathbf{E}[V_{n+1}(\Omega_{n+1})|\Omega_{n}=x] \right\}$ by (\ref{eq:dynamicprogram}),
$V_{n}(x)$ is also convex and increasing.

With the convexity established above, we know that for any $n=1,2,\dots N-1$, $\mathbf{E}[V_{n+1}(\Omega_{n+1})|\Omega_{n}=x]$ and $c_{n} x$ can have at most two intersections because $E[V_{n+1}(\Omega_{n+1})|\Omega_{n}=x]$ is convex and $c_{n} x$ is linear.
We show below that they have exactly one intersection.
For that matter, we want to prove by induction
\begin{align}
 V_n(x)\leq c_n \mathbf{E}[\Omega_N|\Omega_{n}=x] .
\label{eq:value_UB}
\end{align}
Inequality (\ref{eq:value_UB}) is easily true at stage $N$. Suppose (\ref{eq:value_UB}) is true for $n+1$.
At stage $n$, since $c_{n} > c_{n+1}$ we obtain
\begin{align}
V_{n}(x)
=&
\max\left\{c_{n} x, \mathbf{E}[V_{n+1}(\Omega_{n+1})|\Omega_{n}=x] \right\} \nonumber\\
\leq &\max\left\{c_{n} x, \mathbf{E}[c_{n+1}\mathbf{E}[\Omega_{N}|\Omega_{(n+1)}]|\Omega_{n}=x] \right\} \nonumber\\
= &\max\left\{c_{n} x, c_{n+1}\mathbf{E}[\Omega_N|\Omega_{n}=x] \right\} \nonumber\\
\leq &\max\left\{c_{n} x, c_{n}\mathbf{E}[\Omega_N|\Omega_{n}=x]  \right\} \nonumber\\
= & c_n \mathbf{E}[\Omega_N|\Omega_{n}=x].
\end{align}
Then inequality (\ref{eq:value_UB}) is true at any stage.
Consequently, we get, as $x\rightarrow\infty $,
\begin{align}
\frac{V_n(x)}{c_{n} x }
\leq& \frac{c_n \mathbf{E}[\Omega_N|\Omega_{n}=x] }{c_{n} x } \nonumber\\
= &\mathbf{E}[\max(1,\frac{\max_{k>n}\omega_{k}}{x} )]\nonumber\\
\leq &\mathbf{E}[1+\frac{\max_{k>n}\omega_{k}}{x}]  \rightarrow 1,
 \label{eq:limub}
\end{align}
where the convergence in (\ref{eq:limub}) is true because of the following.
Note that $\frac{\max_{k>n}\omega_{k}}{x}$ converges point-wise to $0$, and
$\frac{\max_{k>n}\omega_{k}}{x} \leq \max_{k>n}\omega_{k}$ for $x \geq 1$.
Since $\max_{k>n}\omega_{k}$ is in $L_1$, by dominated convergence theorem $\frac{\max_{k>n}\omega_{k}}{x}$ converges to $0$ in expectation.

When $x =0$ we have
\begin{align}
c_{n} x = 0 \leq \mathbf{E}[V_{n+1}(\omega_{n+1})]=\mathbf{E}[V_{n+1}(\Omega_{n+1})|\Omega_{n}=0].
 \label{eq:0ub}
\end{align}
From (\ref{eq:limub}) and (\ref{eq:0ub}), the linear function $c_{n} x $ is above $\mathbf{E}[V_{n+1}(\Omega_{n+1})|\Omega_{n}=x]$ for large $x$, and
$c_{n} x $ is below $\mathbf{E}[V_{n+1}(\Omega_{n+1})|\Omega_{n}=x]$ when $x=0$.
Therefore, the two functions
$c_{n} x $ and $\mathbf{E}[V_{n+1}(\Omega_{n+1})|\Omega_{n}=x]$ have exactly one intersection.

Let $t_n$ denote the intersection of $c_{n} x $ and $\mathbf{E}[V_{n+1}(\Omega_{n+1})|\Omega_{n}=x]$, we get
\begin{align}
V_{n}(x) = \left\{ \begin{array}{ll}
	c_{n} x & \text{ if }x \geq t_n, \\
	\mathbf{E}[V_{n+1}(\Omega_{n+1})|\Omega_{n}=x] & \text{ if }x < t_n.
\end{array}\right.
\label{eq:valuefunction_thapp}
\end{align}
As a result of (\ref{eq:valuefunction_thapp}), the optimal stopping time is described by (\ref{eq:optimalstopping}).
\end{proof}

\section{Proof of \eqref{eq:upperboundn} in Theorem \ref{thm:fullorder}}\label{pf:fullorder}

Let $\Phi(\gamma)$ be the symbol error rate function with SNR equals $\gamma$ for each node.
Since we apply perfect DF protocol, there are two cases to consider with.
In the first case the relay perfectly reconstructs the signal and the destination receives both signals from the source and the relay.
In the second case the relay fails to reconstruct the signal, so only the source transmitted signal is received at the destination.
The probability for the first case is $1-\Phi(\gamma_{k})$ and the one for the second case is $\Phi(\gamma_{k})$
where $\gamma_{k} = \frac{P_s \omega_{s,k}}{\eta_0} = r\omega_{s,k}c_n\gamma$ is the SNR at $R_k$.
The error rate at the destination $D$ is given by $\Phi(\gamma_{d})$, where $\gamma_{d}$ is the SNR at $D$. $\gamma_{d}$ is can be computed as
\begin{align}
\gamma_d
		=& \left\{
\begin{array}{ll}
	(r\omega_{s,d}+(1-r)\omega_{k,d})c_n\gamma & \text{ in the first case,}\\
	 r\omega_{s,d}c_n\gamma & \text{ in the second case.}
\end{array} \right.
\end{align}
Let $\mathbf{P}_{e,n}(\gamma,\omega_{s,k},\omega_{s,d},\omega_{k,d})$ be the error probability conditional on the channel parameters $\omega_{s,k},\omega_{s,d},\omega_{k,d}$ and the relay selection stops at stage $n$ and $R_k$ is selected.
The error probability can be calculated as follows.
\begin{align}
	&\mathbf{P}_{e,n}(\gamma,\omega_{s,k},\omega_{s,d},\omega_{k,d}) \nonumber\\
    = &(1-\Phi(r\omega_{s,k}c_n\gamma)) \Phi((r\omega_{s,d}+(1-r)\omega_{k,d})c_n\gamma) \nonumber\\
     & ~~~~~~+\Phi(r\omega_{s,k}c_n\gamma) \Phi(r\omega_{s,d} c_n\gamma)~.
    \label{oppur1eq2}
\end{align}
The error probability calculated in (\ref{oppur1eq2}) is complex.
However, a proper upper bound can allow us to analyze the diversity gain of our relay selection strategy.
We proceed to obtain an upper bound on the error probability.

Note that from (\ref{eq:omega_n}), the definition of $\omega_k$, we have
\begin{eqnarray}
\omega_{s,d} \geq \frac{\omega_k}{2q_2},
~\omega_{k,d} \geq \frac{\omega_k}{2q_1}~.
\end{eqnarray}
Following which we strike an upper bound for (\ref{oppur1eq2}) as follows.
\begin{align}
	&\mathbf{P}_{e,n}(\gamma,\omega_{s,k},\omega_{s,d},\omega_{k,d}) \nonumber\\
	= &(1-\Phi(r\omega_{s,k}c_n\gamma)) \Phi((r\omega_{s,d}+(1-r)\omega_{k,d})c_n\gamma) \nonumber\\
     & +\Phi(r\omega_{s,k}c_n\gamma) \Phi(r\omega_{s,d} c_n\gamma) \nonumber\\
    \leq & \Phi((1-r)\omega_{k,d}c_n\gamma)  +\Phi(r\omega_{s,k}c_n\gamma)\nonumber\\
     \leq & \Phi(\frac{1-r}{2q_1}\omega_kc_n\gamma) + \Phi(\frac{r}{2q_2} \omega_kc_n\gamma) \nonumber\\
     \leq & A_M\exp(-B_M\frac{1-r}{2q_1}\omega_kc_n\gamma) + A_M\exp(-B_M\frac{r}{2q_2} \omega_kc_n\gamma)~.
\label{eq:upperboundk}
\end{align}
where $A_M,B_M$ are constants depending on the modulation scheme as commonly adopted (e.g., \cite{proakis2001digital}) and the last inequality in (\ref{eq:upperboundk}) follows from the properties of the error probability function $\Phi(\gamma)$.

From (\ref{eq:upperboundk}) we have
\begin{align}
 &\mathbf{P}_{e,n}(\gamma,x) \leq A_M\exp(-B_M\frac{1-r}{2q_1}xc_n\gamma) \nonumber \\
 &~~~+ A_M\exp(-B_M\frac{r}{2q_2} xc_n\gamma) := Q_n(\gamma,x)~
\end{align}
and \eqref{eq:upperboundn} is established.

\end{appendices}

\end{document}